\documentclass[11pt]{article}

\usepackage{dsfont}

\usepackage{tikz}
\usepackage{pgfplots}
\pgfplotsset{compat = newest}
\usetikzlibrary{decorations.pathreplacing,calligraphy}

\usepackage{amssymb,amsmath,amsthm}
\usepackage{bbm}
\usepackage[noend]{algpseudocode}
\usepackage{graphicx}
\usepackage{verbatim}%for comment blocks
\usepackage{url,xspace,hyperref}
\usepackage{makecell}
\usepackage{cite}
\usepackage{caption}
\usepackage{subcaption}
\usepackage{multirow}
\usepackage{multicol}
\usepackage{latexsym}
\usepackage{amsmath,amssymb,enumerate}
\usepackage{xspace}
\usepackage{algorithm,algorithmicx}
\usepackage{float}
\usepackage{xcolor}
\usepackage{mathrsfs}
\usepackage{cleveref}
\usepackage{cancel}

\usepackage{fullpage}
\usepackage{geometry}
\usepackage{setspace}

\newcounter{note}[section]

\allowdisplaybreaks

\def \R{\mathbb R}

\def \oo{\Tilde{o}}
\def \aa{\Tilde{a}}
\def \sscore{{\widetilde{score}}}

\def\sse{\subseteq}

\def \E{\mathbb{E}}

\newcommand{\dom}{\mathsf{dom}}

\def\asc{\ensuremath{{\sf ASC}}\xspace}
\def\topt{\overline{\opt}}

\DeclareMathOperator*{\argmax}{arg\,max}

\numberwithin{equation}{section}

\theoremstyle{plain}
\newtheorem{thm}{Theorem}[section]

\newtheorem{lem}[thm]{Lemma}

\theoremstyle{definition}
\newtheorem{defn}[thm]{Definition}

\theoremstyle{remark}

\theoremstyle{plain}

\theoremstyle{remark}

\theoremstyle{plain}

\newcommand{\ignore}[1]{}

\title{Minimum Cost Adaptive Submodular Cover  }

\author{Hessa Al-Thani \and Yubing Cui\and Viswanath Nagarajan\thanks{Department of Industrial and Operations Engineering, University of Michigan, Ann Arbor, USA. Research supported in part by NSF grants CMMI-1940766 and CCF-2006778.}}
\begin{document}
\date{}
\maketitle

\begin{abstract}
Adaptive submodularity is a fundamental concept in stochastic optimization, with numerous applications such as sensor placement, hypothesis identification and viral marketing. We consider the problem of minimum cost cover of adaptive-submodular functions, and provide a  $4(1+\ln Q)$-approximation algorithm, where $Q$ is the goal value.  In fact, we consider a significantly more general objective of minimizing the $p^{th}$ moment of the coverage cost, and show that our algorithm {\em simultaneously} achieves a $(p+1)^{p+1}\cdot (\ln Q+1)^p$ approximation guarantee for all $p\ge 1$.  
 All our approximation ratios are  best possible up to constant factors (assuming $P\ne NP$).  
Moreover, our results also extend to the setting where one wants to  cover {\em multiple} adaptive-submodular functions.   Finally, we evaluate the empirical performance of our algorithm on instances of  hypothesis identification.

%Previously,  $O(\ln Q)$-approximation algorithms were only known in  special cases (independent items or unit-cost items).

\end{abstract}
 
\section{Introduction}

Adaptive stochastic optimization, where an algorithm makes sequential decisions while (partially) observing uncertainty, arises in numerous applications such as active learning \cite{D01}, sensor placement \cite{GuestrinKS05} and viral marketing \cite{Tong+17}. Often, these
applications involve an underlying submodular function, and the framework of adaptive-submodularity (introduced by \cite{GK11}) has been widely used to solve these problems. In this paper, we study a basic problem in this context:  covering an adaptive-submodular function at the minimum expected cost.  

In some applications, such as sensor placement (or stochastic set cover \cite{GV06}), the uncertainty can be captured by an {\em independent} random variable associated with each decision. However, there are also a number of applications where the random variables associated with different decisions are correlated. The adaptive-submodularity framework that we consider is also applicable in certain applications  involving correlations.

One such application  is  the viral marketing problem, where we are given a social network and target $Q$, and the goal is to influence at least $Q$ users to adopt a new product. A user can be influenced in two ways (i) directly because the user 
is offered a promotion, or (ii) indirectly because some friend of the user was influenced {\em and} the friend  influenced this user. We incur a cost only in case (i), which accounts for the promotional offer. A widely-used model for  influence behavior is the {\em independent cascade model}~\cite{KempeKT15}. Here, each arc $(u,v)$ has a value $p_{uv}\in [0,1]$ that represents the probability that user $u$ will influence user $v$ (if $u$ is already influenced). A solution is  a sequential process that in each step, selects one user $w$ to influence directly, after which we  get to observe which of $w$'s friends were influenced (indirectly),  and which of their friends were influenced, and so on. So, the solution can utilize these partial observations to make decisions {\em adaptively}. 
Such an adaptive solution can be represented by a decision tree; however, it may require exponential space to store explicitly. We will analyze simple solutions that can be implemented in polynomial time (and space), but our performance guarantees  are relative to an optimal solution that can be very complex. Also, note that the random observations associated with different decisions (in the viral marketing problem) are highly correlated:   the set of nodes that get (indirectly) influenced by any node $w$ depends on the entire network (not just $w$).

While there has been extensive prior work on minimum cost cover of adaptive submodular functions~\cite{GV06,GK11,INZ12,HellersteinKP21,EKM21},  all these results   focus on minimizing the expected cost, which is a risk-neutral objective. However, one may also be interested in minimizing a higher moment of the random cost, which corresponds to a {\em risk-averse} objective. We note that the quality of a solution may vary greatly depending on the chosen objective. For example, consider two solutions $A$ and $B$. Solution $A$ has cost $1$ with probability (w.p.) $1-\frac1M$ and cost $M$ w.p. $\frac1M$. Solution $B$ has cost $M^{1/3}$  w.p.  $1$. The expected cost (i.e., first moment) of $A$ is $2-\frac1M$, whereas that of $B$ is $M^{1/3}$. On the other hand, the second moment of $A$ is $\approx M$, whereas that of $B$ is $M^{2/3}$. Clearly, solution $A$ is much better in terms of the expected cost, whereas solution $B$ is much better in terms of the second moment. 

Motivated by this, we consider the adaptive submodular cover problem under the more general objective of minimizing the $p^{th}$ moment  cost, for any $p\ge 1$. Somewhat surprisingly, we show that there is a  {\em universal}  algorithm for adaptive-submodular cover that approximately minimizes  all moments simultaneously. We note that our result is the first approximation algorithm for higher moments ($p>1$), even for widely studied special cases such as (independent) stochastic submodular cover~\cite{INZ12,HellersteinKP21} and optimal decision tree~\cite{AH12,GuilloryB09,GNR17}.

\subsection{Problem Definition}\label{sec:Definition}

\paragraph{Random items.} Let  $E$ be a finite set of $n$ {\em items}. Each item $e \in E$ corresponds to a random variable $\Phi_e \in \Omega$, where $\Omega$ is the {\em outcome space} (for a single item). We use $\Phi= \langle \Phi_e : e\in E\rangle$ to denote the vector of all random variables (r.v.s). The r.v.s may be arbitrarily correlated across items.   
We use upper-case letters to represent r.v.s and the corresponding lower-case letters to represent realizations of the r.v.s. Thus, for any item $e$, $\phi_e\in \Omega$ is the realization of $\Phi_e$; and $\phi=\langle \phi_e : e\in E\rangle$ denotes the realization  of $\Phi$.
 Equivalently,  we can represent the realization $\phi$  as a subset $\{(e,\phi_e): e\in E\}\sse E\times \Omega$ of item-outcome pairs.

 A {\em partial realization} $\psi\sse E\times \Omega$ refers to the realizations of any {\em subset} of items;  $\dom(\psi)\sse E$ denotes the items whose realizations are represented in $\psi$, and $\psi_e$ denotes the realization of any item $e\in \dom(\psi)$. Note that a partial realization contains at most one pair of the form $(e,*)$ for any item $e\in E$. The (full) realization $\phi$ corresponds to a partial realization with $\dom(\phi)=E$. For two partial realizations $\psi,\psi'\sse E\times \Omega$, we say that $\psi$ is a \textit{subrealization} of $\psi{'}$ 
 (denoted $\psi\preccurlyeq\psi{'}$)  
 if $\psi\subseteq\psi'$; in other words, $\dom(\psi)\sse \dom(\psi')$ and $\psi_e=\psi'_e$ for all $e\in \dom(\psi)$.\footnote{We use the notation $\psi\preccurlyeq\psi{'}$ instead of $\psi\sse\psi{'}$ in order to be consistent with prior works.} 
 Two partial realizations $\psi,\psi'\sse E\times \Omega$ are said to be {\em disjoint} if there is no full realization $\phi$ with $\psi \preccurlyeq\phi$ and $\psi' \preccurlyeq\phi$; in other words, there is some item $e\in \dom(\psi)\cap \dom(\psi')$ such that the realization of $\Phi_e$ is different under $\psi$ and $\psi'$.
 
 We assume that there is a prior probability distribution $p(\phi)=\Pr[\Phi=\phi]$ over realizations $\phi$. Moreover, for any partial realization $\psi$, we assume that we can compute the posterior distribution $p(\phi|\psi) = \Pr(\Phi=\phi|\psi\preccurlyeq\Phi).$

\paragraph{Utility function.} 
 In addition to the random items (described above), there is a {\em utility function} $f : 2^{E\times\Omega} \to \R_{\geq 0}$ that assigns a value to any partial realization.  We will assume that this function is monotone, i.e., having more realizations can not reduce the value. Formally, 
 \begin{defn}[Monotonicity]\label{Definition:mon}
A function $f : 2^{E\times\Omega} \to \R_{\geq 0}$ is {\bf monotone} if 
$$f(\psi)\le f(\psi')\quad \mbox{ for all  partial realizations }\psi\preccurlyeq\psi{'}.$$
 \end{defn}

We also assume that the function $f$ can always achieve its maximal value, i.e.,

\begin{defn}[Coverable]\label{Definition:coverable}
Let $Q$ be the maximal value of function $f$. Then, function $f$ is said to be {\bf coverable} if this value $Q$ can be achieved under every (full) realization, i.e., 
\[f(\phi)=Q \mbox{ for all realizations $\phi$ of  }\Phi.\]
\end{defn}

Furthermore, we will assume that the function $f$ along with the probability distribution $p(\cdot)$ satisfies a submodularity-like property. Before formalizing this, we need the following definition. 
\begin{defn}[Marginal benefit]  The {\bf conditional expected marginal benefit}  of an item $e\in E$ conditioned on observing the partial realization $\psi$ is:
\[\Delta(e|\psi):=\E\left[f(\psi \cup (e, \Phi_e)) -f( \psi) \,|\, \psi\preccurlyeq\Phi\right] \,\,=\,\, \sum_{\omega\in \Omega}\Pr[\Phi_e=\omega|\psi\preccurlyeq\Phi]\cdot \left( f(\psi\cup (e,\omega) )-f(\psi)\right).\]
\end{defn}

We will assume that function $f$ and distribution $p(\cdot)$ jointly satisfy the adaptive-submodularity property, defined as follows. 
\begin{defn}[Adaptive submodularity]\label{Definition:adsub}
        A function $f : 2^{E\times\Omega} \to \R_{\geq 0}$ is {\bf adaptive submodular} w.r.t. distribution $p(\phi)$ if for all partial realizations $\psi\preccurlyeq\psi^{'}$, and all items $e \in E\setminus\dom(\psi^{'})$, we have
\[\Delta(e|\psi)\geq \Delta(e|\psi^{'}).\]
\end{defn}

In other words, this property ensures that the marginal benefit of an  item never increases as we condition on more realizations. 
Given any function $f$ satisfying Definitions~\ref{Definition:mon}, \ref{Definition:coverable} and \ref{Definition:adsub},  we can pre-process $f$ by subtracting $f(\emptyset)$, to get an equivalent function (that maintains these properties), and has a smaller $Q$ value. So, we may assume that $f(\emptyset)=0$.  

\def\cp{c_{p}}

\paragraph{Min-cost adaptive-submodular cover (\asc).} 

In this problem, each item $e\in E$ has a positive cost $c_e$. The goal is to select items (and observe their realizations)  sequentially until the observed realizations have function value $Q$.  The objective is to minimize the expected cost of selected items. 

Due to the stochastic nature of the problem, the solution concept here is much more complex than in the deterministic setting (where we just select a static subset). In particular, a solution  corresponds to a ``policy'' that maps observed realizations to the next selection decision. 
The observed realization at any point corresponds to a partial realization (namely, the realizations of the items selected so far). 
 Formally, a {\em policy} is a mapping $\pi:2^{E\times\Omega}\to E$, which specifies the next item $\pi(\psi)$ to select when the observed realizations are $\psi$.\footnote{Policies  and utility functions  are not necessarily defined over all subsets $2^{E\times \Omega}$, but only over partial realizations;  recall that a  partial realization is of the form $\{(e,\phi_e): e\in S\}$ where $\phi$ is some full-realization and $S\sse E$.}  
The policy $\pi$ terminates at the first point when $f(\psi)=Q$, where $\psi\sse E\times \Omega$ denotes the observed realizations so far. For any policy $\pi$ and full realization $\phi$, let $C(\pi , \phi )$ denote the total cost of items selected by policy $\pi$ under realization $\phi$. Then, the expected cost of   policy $\pi$ is:
$$c_{exp}(\pi) \,\,=\,\, \E_{\Phi}\left[C(\pi, \Phi)\right] \,\,=\,\, \sum_\phi p(\phi)\cdot C(\pi, \phi).$$
While minimizing the expected cost is the primary objective, we are also interested in minimizing higher moments of the cost. For any $p\ge 1$ and policy $\pi$, the $p^{th}$ moment 
of the policy's cost is:
$$\cp(\pi) \,\,=\,\, \E_{\Phi}\left[C(\pi, \Phi)^p\right] \,\,=\,\, \sum_\phi p(\phi)\cdot C(\pi,\phi)^p.$$

At any point in policy $\pi$, we refer to the cumulative cost incurred so far as the {\em time}. If $J_1, J_2,\cdots J_k$ denotes the (random) sequence of items selected by $\pi$ then for each $i\in \{1,2,\cdots k\}$, we view item $J_i$ as being selected during the time interval $[\sum_{h=1}^{i-1} c(J_h) \, ,\, \sum_{h=1}^{i} c(J_h) )$ and the realization of $J_i$ is only observed at time $\sum_{h=1}^{i} c(J_h)$. For any time $t\ge 0$, we use $\Psi(\pi,t)\sse E\times \Omega$ to denote the (random) realizations that have been observed by time $t$ in policy $\pi$. We note that $\Psi(\pi,t)$ only contains the realizations of items that have been {\em completely} selected by time $t$. Note that the policy terminates at the earliest time $t$ where $f(\Psi(\pi,t))=Q$.

Given any policy $\pi$, we define its \textit{cost $k$ truncation}   by running $\pi$ and stopping it just before the cost of selected items exceeds $k$. That is, we stop the policy as late as possible while ensuring that the cost of selected items never exceeds $k$ (for any realization).

\paragraph{Remark:} Our definition of the utility function $f$ is slightly more restrictive than the original definition \cite{GK11}. In particular, the utility function in \cite{GK11} is of the form $g:2^E\times \Omega^E\rightarrow \R_{\ge 0}$, where the function value $g(\dom(\psi), \Phi)$ for any partial realization $\psi$ is still random and can depend on the outcomes of unobserved items, i.e., those in $E\setminus \dom(\psi)$. Nevertheless, our formulation (\asc) still captures most applications of the formulation  studied in \cite{GK11}. See Section~\ref{sec:appln} for details.

\def\grd{\pi}
\def\opt{\sigma}
\subsection{Adaptive Greedy Policy}
Algorithm~\ref{alg:THE_algo} describes a natural greedy policy for min-cost adaptive-submodular cover, which has also been studied in prior works \cite{GK17,EKM21,HellersteinKP21}.

\begin{algorithm}[h!]
\caption{Adaptive Greedy Policy $\grd$.\label{alg:THE_algo}}
\begin{algorithmic}[1]
\State selected items $A \gets \emptyset$, observed realizations $\psi \gets \emptyset$
\While{$f(\psi)<Q$}
\State $e^{*} = \argmax_{e\in E\setminus A}\frac{\Delta(e|\psi)}{c_e}$
\State add $e^*$ to the selected items, i.e., $A \gets A\cup\{e^{*}\}$
\State select $e^*$ and observe $\Phi_{e^*}$
\State update $\psi \gets \psi\cup\{(e^{*}, \Phi_{e^{*}})\}$
\EndWhile
\end{algorithmic}
\end{algorithm}
\paragraph{Remark:}
Note that the policy $\grd$ remains the same if we replace the greedy choice by 
\begin{equation}
\label{eq:greedy-choice}
e^{*} = \argmax_{e\in E\setminus A}\frac{\Delta(e|\psi)}{c_e\cdot (Q-f(\psi))}.
\end{equation}
This is because the additional term $Q-f(\psi)$ is the same for each item $e\in E\setminus A$ (note that at any particular step, $\psi$ is a fixed partial realization). We will make use of this alternative greedy criterion in our analysis.

\subsection{Results and Techniques}
Our first main result is on the expected cost of the greedy policy.
\begin{thm}\label{Theorem: THE_thm}
Consider any instance of minimum cost adaptive-submodular cover, where the utility function $f:2^{E\times\Omega}\to\R_{\ge 0}$ is monotone, coverable and adaptive-submodular w.r.t. the probability distribution $p(\cdot)$. Suppose that there is some value $\eta>0$ such that $f(\psi)>Q-\eta$ implies $f(\psi)=Q$ for all partial realizations $\psi\sse E\times \Omega$. Then, the expected cost of the greedy policy is 
\[c_{exp}(\grd)\leq 4\cdot (1+\ln (Q/\eta))\cdot  c_{exp}(\opt),\]
where $\opt$ denotes the optimal policy.
\end{thm}

 This is an asymptotic improvement over the $(1+ \ln (Q/\eta))^2$-approximation bound from \cite{GK17} and the $(1+\ln (\frac{nQc_{max}}\eta))$-approximation bound from \cite{EKM21};   the maximum item cost $c_{max}$ can even be exponentially larger than $Q$. Our bound is the best possible, up to the constant factor of $4$, because the set cover problem is a special case of \asc (where $Q$ is the number of elements to cover and $\eta=1$). \cite{DinurS14} showed that,  assuming $P\ne NP$, one cannot obtain a   better than $\ln (Q)$ approximation ratio for set cover.

As a consequence, we obtain the first $O(\ln Q)$-approximation algorithm for the viral marketing application mentioned earlier. We also obtain an improved bound for the optimal decision tree problem with uniform priors. See Section~\ref{sec:appln} for details.

In fact, we obtain Theorem~\ref{Theorem: THE_thm} as a special case of the following general result on minimizing the $p^{th}$ moment of the coverage cost.

\begin{thm}\label{Theorem: THE_thmp}
Consider any instance of minimum cost adaptive-submodular cover, where the utility function $f$ is monotone, coverable and adaptive-submodular w.r.t. the probability distribution $q(\cdot)$. Suppose that there is some value $\eta>0$ such that $f(\psi)>Q-\eta$ implies $f(\psi)=Q$ for all partial realizations $\psi\sse E\times \Omega$. Then, for every $p\ge 1$, the $p^{th}$ moment   cost of the greedy policy is
\[\cp(\grd)  \leq (p+1)^{p+1}\cdot (1+\ln (Q/\eta))^p \cdot  \cp (\opt_p) ,\]
where $\opt_p$ denotes the optimal policy for the $p^{th}$ moment objective.
\end{thm}

This result is also best possible for any $p\ge 1$, up to the leading constant of $(p+1)^{p+1}$. We emphasize that  the greedy policy $\pi$ is   oblivious  to the choice of $p$: so Theorem~\ref{Theorem: THE_thmp} provides a universal algorithm that achieves the stated approximation ratio for {\em every} value of  $p$. We are not aware of any previous result for minimizing higher moments, even for  special cases such as (independent) stochastic submodular cover or optimal decision tree. 

Our proof technique  is  very different from prior works on adaptive submodular cover~\cite{GK17,EKM21,HellersteinKP21}. The approaches in these papers  were tailored to bound the expected cost, and seem difficult to extend to higher moments $p>1$. We start by expressing the $p^{th}$ moment objective of any policy as an appropriate integral of  ``non-completion probabilities'' over all times. Then, we relate the non-completion probabilities in the  greedy policy (at any time $t$) to that in the optimal policy (at a scaled time $\frac{t}\alpha$). 
 In order to establish this relation, we consider the integral of the   greedy criterion value over each suffix $(t,\infty)$ of time,  and prove lower and upper bounds on this quantity. Finally, we relate the $p^{th}$ moment objectives of greedy and the optimal policy by analyzing a double integral over all suffixes.   
The high-level approach of using non-completion probabilities has been used earlier for minimizing the expected cost in a number of stochastic covering problems, including the independent special case of \asc \cite{INZ12}. We refine and improve this approach by (i) utilizing  non-completion probabilities at all times (not just  power-of-two times) and (ii) using a stronger upper bound on the greedy criterion value.  Moreover, we extend this approach to $p^{th}$ moment objectives (prior work only looked at expected cost).  
 
Additionally, our algorithm and analysis extend in a straightforward manner, to the setting 
with 
multiple adaptive-submodular functions, where the objective is the sum of $p^{th}$ moments of the  ``cover times'' of all the functions. We obtain the same approximation ratio even for this more general problem. 
 In fact, the multiple \asc problem with $Q=\eta=1$ and an  expected cost objective (i.e., $p=1$) generalizes the min-sum set cover problem~\cite{FLT04}, which is NP-hard to approximate better than factor $4$. The  constant factor $(p+1)^{p+1}$ in our approximation ratio for this problem is also $4$, which implies  that our bound is  tight in this case.

Finally, we provide computational results of our algorithm on real-world instances of optimal decision tree. We compare the performance of our algorithm on $p^{th}$  moment objectives for   $p=1,2,3$ to lower-bounds that we obtain (via Huffman coding). Our algorithm performs very well on the instances tested. 

\subsection{Related Work}
%adp-sub
Adaptive submodularity was introduced by \cite{GK11}, where they considered both the maximum-coverage  and the minimum-cost-cover problems. They showed that the greedy policy is a $(1-\frac1e)$ approximation for maximum coverage, where the goal is to maximize the expected value of  an adaptive-submodular function subject to a cardinality constraint. They also claimed that the greedy policy is a 
$(1+ \ln (Q/\eta))$ approximation for min-cost cover of an adaptive-submodular function. However, this result had an error \cite{NS17}, and a corrected proof \cite{GK17} only provides a double-logarithmic  $(1+ \ln (Q/\eta))^2$ approximation. Recently, \cite{EKM21} obtained a single-logarithmic approximation bound of $(1+ \ln (\frac{n Q c_{max}}\eta))$. However, this bound depends additionally on the number of items $n$ and their maximum cost $c_{max}$. Our result shows that the greedy policy is indeed an $O(\ln (Q/\eta))$ approximation. As noted earlier, our definition of \asc is simpler and slightly more restrictive than the original one in \cite{GK11}, although most applications of adaptive-submodularity do  satisfy our definition.

%indep case
The special case of adaptive-submodularity where the random variables are independent across items, has also been studied extensively. For the maximum-coverage version, \cite{AsadpourN16} obtained  a $(1-\frac1e)$-approximation algorithm via a ``non adaptive'' policy (that fixes a subset of items to select upfront). Subsequent work \cite{GuptaNS17,BSZ19,AdamczykSW16} obtained   constant factor approximation algorithms for a variety of constraints (beyond just cardinality). The minimum-cost cover problem (called {\em stochastic submodular cover}) was studied in \cite{AAK19,HellersteinKP21,HellersteinK18,GhugeGN21,INZ12}. In particular, an $O(\ln (Q/\eta))$ approximation algorithm follows from \cite{INZ12}, and 
recently \cite{HellersteinKP21} proved that the greedy policy has a $(1+ \ln (Q/\eta))$ approximation guarantee. The latter guarantee is the best possible, even up to the constant factor: this also matches the best approximation ratio for the {\em deterministic} submodular cover problem~\cite{W82} and its special  case of set cover~\cite{DinurS14}. The papers \cite{AAK19,GhugeGN21} studied stochastic submodular cover under limited ``rounds of adaptivity'', and obtained smooth taredoffs between the approximation ratio and the number of rounds.

%multiple fns
The (deterministic) submodular cover problem with multiple functions was introduced  in \cite{AzarG11}, where they obtained  an $O(\ln (Q/\eta))$ approximation algorithm for minimizing the sum of cover times. Subsequently,  \cite{INZ12} studied the {\em stochastic} submodular cover problem with multiple functions (which involved independent items), and obtained  an $O(\ln (Q/\eta))$ approximation algorithm. The analysis in our paper is similar, at a high level, to the analysis in \cite{INZ12}, which also relied on  the non-completion probabilities. However, we handle the more general adaptive-submodular setting (where items may be correlated), and we obtain a much better constant factor and extend the techniques to minimizing higher moments of the cost.

Minimizing higher moments of the covering cost has been studied previously in the  {\em deterministic} setting by \cite{GolovinGKT08}. Specifically, they considered the {\em $L_p$  set cover} problem, where given a collection of sets (items in our setting) and elements, the goal is to find a sequence of sets so as to minimize the total $p^{th}$ moment of   cover times over all elements. \cite{GolovinGKT08} showed that the greedy algorithm  for $L_p$  set cover achieves an approximation ratio of $(p+1)^{p+1}$ for each $p\ge 1$ simultaneously.  We note that $L_p$  set cover   is a special case of the multiple \asc problem studied in our paper, where  $Q=\eta=1$ and each function corresponds to an element in set cover. So, our approximation ratio for multiple \asc in this special case matches the best known bound for $L_p$  set cover. Moreover, \cite{GolovinGKT08} showed that for any fixed value of $p$, there is no approximation ratio better than $\Omega(p)^p$ for $L_p$  set cover, unless $NP\sse DTIME(n^{\log\log n})$. So, the leading factor $(p+1)^{p+1}$ in our bound for multiple \asc is nearly the best possible, for each $p$.   As noted before, ours is the first paper to consider higher moment objectives in the {\em stochastic} setting, even for special cases of \asc such as stochastic submodular cover and optimal decision tree. Our proof technique is also very different from that in \cite{GolovinGKT08}.

A different (scenario based) model for correlations in adaptive submodular cover was studied in \cite{NavidiKN20,GHKL16}. Here, the utility function $f$ is just required to be submodular (not adaptive-submodular), but the algorithm requires an explicit description of the probability distribution $p(\cdot)$.  In particular, \cite{NavidiKN20} obtained a greedy-style policy with approximation ratio $O(\ln(mQ/\eta))$ where $m$ is the support-size of  distribution $p(\cdot)$, and $Q$ and $\eta$ are as before.  
We believe that our proof technique may be useful in extending the expected-cost   results in \cite{NavidiKN20} to higher moments  and in  improving the constant factor in their approximation ratio.

\section{Analyzing the Greedy Policy} 

In this section, we prove our main result (Theorem~\ref{Theorem: THE_thmp}), which bounds the $p^{th}$ moment of the greedy policy cost. Note that  setting $p=1$ in Theorem~\ref{Theorem: THE_thmp} implies the bound on expected cost (Theorem~\ref{Theorem: THE_thm}). We first show how to re-write the $p^{th}$ moment objective in terms of  ``non-completion'' probabilities. Then, we relate the non-completion probabilities in the greedy and optimal  policies by analyzing the integral of the ``greedy criterion value'' \eqref{eq:greedy-choice} over time.

\subsection{Non-completion Probabilities and the Moment Objective}
Our analysis is based on relating the ``non-completion'' probabilities at different times in the greedy policy $\grd$ and the optimal policy $\opt$. We first define these quantities formally.

\begin{defn}[Non-completion probabilities] For any time $t\ge 0$, let 
    \begin{equation*}
        o(t)\,:= \, \Pr[\opt \mbox{ does not terminate by time }t] = \Pr[C(\opt, \Phi) > t] = \Pr[f(\Psi(\opt , t)) < Q] .
    \end{equation*}
    Similarly, for any $t\ge 0$, let
    \begin{equation*}
        a(t)\,:= \, \Pr[\grd \mbox{ does not terminate by time }  t] = \Pr[C(\grd, \Phi) > t] = \Pr[f(\Psi(\grd ,  t)) < Q] .
    \end{equation*}
    \label{definition:noncomp}
    \end{defn}
    
See Figure~\ref{fig:example_function_o} for an example of the non-completion probabilities $o(t)$. Clearly,   $o(t)$ and $a(t)$ are non-increasing functions of $t$. Moreover, $o(0)=a(0)=1$ and we have $o(t)=a(t)=0$ for all $t\ge \sum_{e\in E} c_e$ (this is because any policy would have selected all items by this time).

It is easy to see that the expected cost of any policy is exactly the integral of the non-completion probabilities over time. That is, 
\begin{equation*}
c_{exp}(\opt)=\int_{0}^{\infty}  \Pr[ \opt \mbox{ does not terminate by time }t] dt = \int_{0}^{\infty}   o(t) dt.
\end{equation*}
\begin{equation*}
c_{exp}(\grd)=\int_{0}^{\infty}  \Pr[ \grd \mbox{ does not terminate by time }t] dt = \int_{0}^{\infty}   a(t) dt.
\end{equation*}

\begin{figure}[h]
    \centering
\begin{tikzpicture}
\begin{axis}[
            xmin = 0, xmax = 10,
            ymin = 0, ymax = 1.2,
            xtick distance = 1,
            ytick distance = 1,
            grid = both,
            minor tick num = 1,
            major grid style = {lightgray},
            minor grid style = {lightgray!25},
            width = 0.5\textwidth,
            height = 0.3\textwidth,
            xlabel = {time $t$},
            ylabel = {$o(t)$},
            legend cell align = {left},
        ]
\addplot[ultra thick][domain=0:1] {1};
\addplot[ultra thick][domain=1:2] {0.7};
\addplot[ultra thick][domain=2:3.2] {0.6};
\addplot[ultra thick][domain=3.2:6.5] {0.3};
\addplot[ultra thick][domain=6.5:8] {0.1};
\addplot[ultra thick][domain=8:10] {0};
\addplot[mark=*] coordinates {(0,1)};
\addplot[mark=*,fill=white] coordinates {(1,1)};
\addplot[mark=*] coordinates {(1,0.7)};
\addplot[mark=*,fill=white] coordinates {(2,0.7)};
\addplot[mark=*] coordinates {(2,0.6)};
\addplot[mark=*,fill=white] coordinates {(3.2,0.6)};
\addplot[mark=*] coordinates {(3.2,0.3)};
\addplot[mark=*,fill=white] coordinates {(6.5,0.3)};
\addplot[mark=*] coordinates {(6.5,0.1)};
\addplot[mark=*,fill=white] coordinates {(8,0.1)};
\addplot[mark=*] coordinates {(8,0)};
\end{axis}
\end{tikzpicture}
    \caption{Graph of a simple $o(\cdot)$ function.
    \label{fig:example_function_o} }
\end{figure}
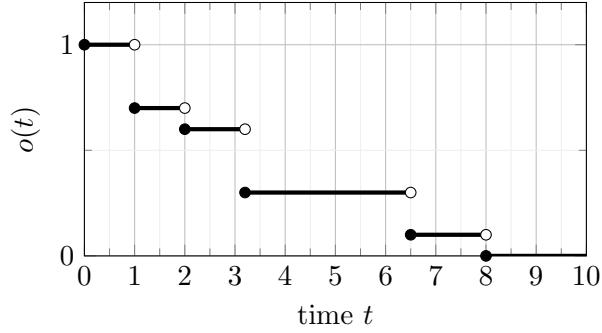

It turns out that we can also express the $p^{th}$ moment objective of any policy as a suitable integral of the non-completion probabilities. This relies on the following result. 
 
\begin{lem}
\label{lem:momentp}
Suppose that $X$ is a non-negative, discrete and bounded random variable.    For any $p \ge 1$, the $p^{th}$ moment of $X$  is
    \[  \E[X^p] = p\int_{t=0}^{\infty}{t^{p-1} \cdot \Pr[X>t]   }dt. \]
\end{lem}
\begin{proof}
Let $V$ denote the (finite) support of r.v. $X$. Note that $x^p = p\cdot \int_{0}^x t^{p-1} dt$ for all $x\ge 0$. So, we obtain
\begin{align*}
     \E[X^p] &  = \sum_{x\in V} \Pr[X=x]\cdot x^p = p \sum_{x\in V} \Pr[X=x] \int_{t=0}^{x} {t^{p-1}dt} \\
     & = p\int_{0}^{\infty}{t^{p-1}  \sum_{x\in V: x > t}{\Pr[ X = x]}dt} = p\int_{0}^{\infty} t^{p-1} \cdot \Pr[X > t] dt.
\end{align*}
The third equality above switches the order of the integral and summation using Fubini's theorem.
\end{proof}

Now, we apply Lemma~\ref{lem:momentp} with random variable $X=C(\sigma,\Phi)$ where $\sigma$ is the  optimal policy and $\Phi$ denotes the r.v.s in the \asc instance.  Note that the cost $C(\sigma,\Phi)$ is non-negative, discrete and bounded. Using the fact that $o(t) = \Pr[C(\opt, \Phi) > t] $,  we  obtain:

\begin{equation}\label{eq:cp-opt}
\cp(\opt)\,\,=\,\, \E_{\Phi}\left[C(\opt, \Phi)^p\right] \,\,=\,\,   p \int_{0}^{\infty}  t^{p-1} \cdot     o(t) dt.
\end{equation}
Similarly, applying Lemma~\ref{lem:momentp} to the r.v. $C(\grd,\Phi)$ for the  greedy  policy $\grd$, 
\begin{equation}\label{eq:cp-grd}
\cp(\grd)\,\,=\,\, \E_{\Phi}\left[C(\grd, \Phi)^p\right] \,\,=\,\,   p \int_{0}^{\infty}  t^{p-1} \cdot   a(t) dt.
\end{equation}

\subsection{Using the Greedy Criterion}

Our analysis of $\grd$ relies on tracking the ``greedy criterion value'' defined in \eqref{eq:greedy-choice}. 
 The following definition formalizes this.  
     
\begin{defn}[Greedy score] 
For $t\ge 0$ and any partial realization $\psi$ observed at time $t$, define 

\begin{equation*}
score(t,\psi) := 
    \begin{cases}
        \frac{\Delta(e|\psi)}{c_e[Q-f(\psi)]}, & \substack{\text{where $e$ is the item being selected in $\grd$ at time $t$} \\ \text{and } \psi 
        \text{ was observed just before selecting } e.}\\
        0, & \text{if no item is being selected  in $\grd$ at time $t$ when }\psi \text{ was  observed}.
    \end{cases}
\end{equation*}
    
Note that conditioned on $\psi$, the item $e$ being selected in $\pi$ at time $t$ is deterministic.   
    \end{defn}
    
The expression for \textit{score} above is exactly the greedy criterion in \eqref{eq:greedy-choice}. Moreover, the score may increase and decrease over time: see Figure~\ref{fig:example_score} for an example.

In order to reduce notation, for any time $t$, we use $\Psi_t:=\Psi(\pi,t)$ to denote the (random) partial realization observed by the greedy policy $\pi$ at time $t$; recall that this only includes items that have been completely selected by time $t$.

\begin{defn}
 (Potential gain)
 For any time  $i\ge 0$, its potential {\bf gain} is the expected total score accumulated after time $i$,  
    \[G_i := \int_{t=i}^{\infty} \E \left[ score(t,\Psi_t)\right] dt.\] 
    \end{defn}

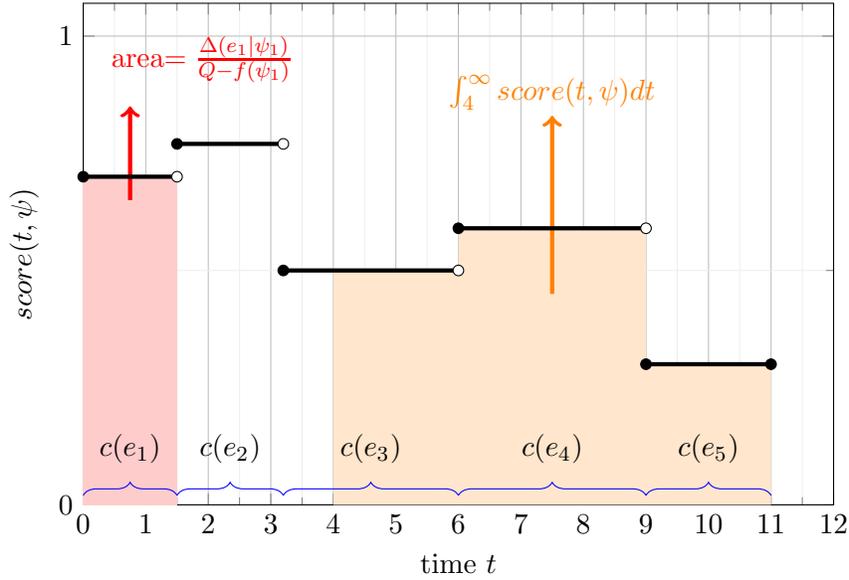
\begin{figure}[h]
    \centering
   
\begin{tikzpicture}
\begin{axis}[
            xmin = 0, xmax = 12,
            ymin = 0, ymax = 1.07,
            xtick distance = 1,
            ytick distance = 1,
            grid = both,
            minor tick num = 1,
            major grid style = {lightgray},
            minor grid style = {lightgray!25},
            width = 0.7\textwidth,
            height = 0.5\textwidth,
            xlabel = {time $t$},
            ylabel = {$score(t,\psi)$},
            legend cell align = {left},
        ]
        
\fill[red!20] (0,0) -- (0,0.7) -- (1.5, 0.7) -- (1.5,0) -- cycle;
\draw[->,red, ultra thick](0.75,0.65)--(0.75,0.85);
\node[red] at (1.9,0.95) {area$=\frac{\Delta(e_1|\psi_1)}{Q-f(\psi_1)}$};
\fill[orange!20] (4,0) -- (4,0.5) -- (6, 0.5) -- (6,0.59) -- (9,0.59) -- (9,0.3) -- (11,0.3) -- (11,0) -- cycle;
\draw[->,orange, ultra thick](7.5,0.45)--(7.5,0.83);
\node[orange, ultra thick] at (7.5,0.88) {$\int_{4  }^{\infty}score(t,\psi)dt$ };
%\node[orange, ultra thick] at (7.5,0.78) { with $L\beta=4$};
\node[orange, ultra thick] at (7.5,0.78) {};
\addplot[mark=*] coordinates {(0,0.7)};
\addplot[ultra thick][domain=0:1.5] {0.7};
\addplot[mark=*,fill=white] coordinates {(1.5,0.7)};
\addplot[mark=*] coordinates {(1.5,0.77)};
\addplot[ultra thick][domain=1.5:3.2] {0.77};
\addplot[mark=*,fill=white] coordinates {(3.2,0.77)}; %white dot
\addplot[mark=*] coordinates {(3.2,0.5)};             %black dot
\addplot[ultra thick][domain=3.2:6] {0.5};
\addplot[mark=*,fill=white] coordinates {(6,0.5)};
\addplot[mark=*] coordinates {(6,0.59)};
\addplot[ultra thick][domain=6:9] {0.59};
\addplot[mark=*,fill=white] coordinates {(9,0.59)};
%\addplot[mark=*] coordinates {(9,0)};
%\addplot[ultra thick][domain=9:12] {0};
\addplot[ultra thick][domain=9:11] {0.3}; %new line from 9-11
\addplot[mark=*] coordinates {(11,0.3)}; %white dot
\addplot[mark=*] coordinates {(9,0.3)};    
\node[orange, ultra thick] at (9,0.88) {};
\draw [decorate, decoration = {brace, amplitude=5pt}, blue] (0,0.02) --  (1.5,0.02);
\node[] at (0.75,0.12) {$c(e_1)$};
\draw [decorate, decoration = {brace, amplitude=5pt}, blue] (1.5,0.02) --  (3.2,0.02);
\node[] at (2.35,0.12) {$c(e_2)$};
\draw [decorate, decoration = {brace, amplitude=5pt}, blue] (3.2,0.02) --  (6,0.02);
\node[] at (4.6,0.12) {$c(e_3)$};
\draw [decorate, decoration = {brace, amplitude=5pt}, blue] (6,0.02) --  (9,0.02);
\node[] at (7.5,0.12) {$c(e_4)$};
\draw [decorate, decoration = {brace, amplitude=5pt}, blue] (9,0.02) --  (11,0.02);
\node[] at (10,0.12) {$c(e_5)$};
\end{axis}
\end{tikzpicture}
    \caption{Graph of a simple $score(t,\psi)$ for illustration. $e_1, e_2,...$ are greedy selections and $\psi_i$ is the partial realization just before selecting $e_i$. }
    \label{fig:example_score}
\end{figure}
    
The key part of the analysis lies in upper and lower bounding the potential gain starting from all time points. The upper-bound relates to the non-completion probabilities $a(\cdot)$ in $\grd$  and the lower-bound relates to the non-completion probabilities $o(\cdot)$ in $\opt$. Putting the upper and lower bound together allows us to relate $a(\cdot)$ and $o(\cdot)$, which in turn will be used   to upper-bound $\cp(\grd)$ in terms of $\cp(\opt)$. 
For the upper-bound (Lemma~\ref{Lemma: upper_bd_lemma}), we view gains as the expected values over full-realizations, which allows us to write the total (conditional) gain as a harmonic series. For the lower-bound (Lemma~\ref{Lemma: low_bd_lemma}), we view the gain as an integral of expected contributions over time and  prove a lower bound for each time step using  the optimal policy $\opt$ and the adaptive-submodularity property.

Below, let $L:=1+\ln (Q/\eta)$ and $\beta> 1$ be some constant value (that will be fixed later).

 \begin{lem}\label{Lemma: upper_bd_lemma}
For any $i\ge 0$, the potential gain at time $i$ is 
\[G_i = \int_{i}^{\infty} \E \left[ score(t,\Psi_t)\right] dt \leq L\cdot a(i).\] 
\end{lem}

\begin{proof}
We start by re-expressing the score and gain in terms of the full realization. For any time $t\geq 0$ and  full  realization $\phi$, let

\begin{equation*}
S(t,\phi) := 
    \begin{cases}
        \frac{f(\psi\cup (e,\phi_e))-f(\psi)}{c_e[Q-f(\psi)]}, & \substack{ \text{where $e$ is the item being selected in $\pi$ at time $t$ under } \phi, \\   \text{and $\psi \preccurlyeq \phi$ is the partial realization just before selecting } e.}\\
        0, & \text{if no item is being selected  in $\grd$ at time $t$ under }\phi.
    \end{cases}
\end{equation*}
    
Then, for any  $t\ge 0$ and any partial realization $\psi$ observed at time $t$,   
$$score(t,\psi) = \E_\Phi[S(t,\Phi) | \psi\preccurlyeq \Phi].$$

This uses the definition of $\Delta(e|\psi)$ and the fact that conditioned on $\psi$, the item $e$ (being selected at time $t$) is fixed.   Hence, for   any time-point $k\ge 0$, its potential gain
$$G_k = \int_{k}^{\infty} \E \left[ score(t,\Psi_t)\right] dt = \int_{k}^{\infty} \E_{\Psi_t} \left[ \E_\Phi[S(t,\Phi)|\Psi_t\preccurlyeq\Phi]\right] dt = \int_{k}^{\infty} \E \left[ S(t,\Phi)\right] dt.$$

Now, fix time $i\ge 0$ and condition on any (full) realization $\phi$.
\begin{enumerate}
    \item[] Case 1: suppose that $\grd$ under $\phi$ terminates before ($<$) time $ i$. Then, $S(t,\phi)=0$ for all $t \ge i$, and so:  
    \begin{equation*}
    \begin{aligned}
        G_i(\phi) = \int_{i}^{\infty}S(t,\phi)dt=0
    \end{aligned}
\end{equation*}

    \item[] Case 2: suppose that $\grd$ under $\phi$ terminates after ($\ge$) time $i$.
\begin{equation*}
    \begin{aligned}
        G_i(\phi)=\int_{i}^{\infty}S(t,\phi)dt=\int_{i}^{\infty}S(t,\phi)dt\leq \int_{0}^{\infty}S(t,\phi)dt\leq  L, 
    \end{aligned}
\end{equation*}
where the last inequality is by Lemma~\ref{Lemma:harmonic} (proved in Appendix~\ref{app}).    
\end{enumerate}

Note that case 2 above happens exactly with probability $a(i)$. So, 
\begin{equation*}
    \begin{aligned}
        G_i=\E_\Phi\left[G_i(\Phi) \, | \, \mbox{case 2 occurs under }\Phi\right]\cdot \Pr[\mbox{case 2 occurs}]\le  L\cdot a(i),    \end{aligned}
\end{equation*}
which completes the proof.\end{proof}

\begin{lem}\label{Lemma: low_bd_lemma}
For any time $t\ge 0$, 
\[\E[score(t, \Psi_t)]\geq \frac{a(t)-o(t/L\beta)}{t/(L\beta)}.\]
Hence, $G_i\ge     \int_{i}^{\infty}\frac{a(t)-o(t/(L\beta))}{t/(L\beta)}dt $ for each time $i\ge 0$.
\end{lem}

\begin{proof} 
  Note that the first statement in the lemma immediately implies the second statement. 
Indeed,
$$G_i = \int_{ i}^{\infty} \E[score(t, \Psi_t)] dt \ge    \int_{ i}^{\infty}\frac{a(t)-o(t/(L\beta))}{t/(L\beta)}dt. $$

We now prove the first statement. 
Henceforth, fix time $t\ge 0$.

\paragraph{Truncated optimal policy}
Let $\topt$ denote the cost of policy $\opt$ truncated at time $ t/(L\beta) $. Note that the total cost of selected items in $\topt$ is always at most $ t/(L \beta)$. However, $\topt$ may not fully cover the utility function $f$ (so it is not a feasible policy for min-cost adaptive submodular cover). We define the following random quantities associated with policy $\topt$:
\begin{equation*}
    \begin{aligned}
    I_k  &:= \text{set of first $k$ items selected by }\topt, \mbox{ for   }k=0,1,\cdots . \\
    I_\infty &:= \text{set of {\em all} items selected by the end of }\topt.\\
    P_k &:=\{(e,\Phi_e):e\in I_k\} \text{, i.e. partial realization of the first $k$ items selected by }\topt , \mbox{ for }k=0,1,\cdots .\\
    P_{\infty}&:=\{(e,\Phi_e):e\in I_\infty\} \text{, i.e. partial realization observed by the end of }\topt. 
    \end{aligned}
\end{equation*}

Note that $\topt$ covers $f$ exactly when $f(P_{\infty})=Q$. Moreover, by definition of the function $o(\cdot)$, we have $\Pr[\topt\text{ covers } f]=1-o(t/(L \beta) )$.
\paragraph{Conditioning on partial realizations in greedy} 
Let $\psi$ be any partial realization corresponding to $\Psi_t$ with $f(\psi)<Q$. In other words, (i) $\psi$ is the partial realization observed at time $t$  in some execution of policy $\grd$, and (ii) the policy has not terminated (under realization $\psi$) by time $t$. Let $R(\grd,t)$ denote the collection of such partial realizations. Note that the partial realizations in $R(\grd,t)$  are mutually disjoint, and the total probability of these partial realizations equals the probability that $\grd$ does not terminate by time $t$.
 We will show that:

\begin{equation}
\label{eq:LB-score-psi}
\Pr[\psi \preccurlyeq \Phi ] \cdot score(t,\psi) \ge \frac{L \beta }{t}\cdot \Pr[(\psi\preccurlyeq\Phi)\wedge(\topt\text{ covers } f)], \qquad \forall \psi\in R(\grd,i).
\end{equation}

We first complete the proof of the lemma assuming \eqref{eq:LB-score-psi}. 

\begin{eqnarray}
        \E[score(t,\Psi_t)]& \ge  &   \sum_{\psi\in  R(\grd,t)}p(\psi)\cdot score(t,\psi) = \sum_{\psi\in  R(\grd,t)}\Pr[\psi \preccurlyeq \Phi ] \cdot score(t,\psi) \notag \\
         &\geq &\frac{L \beta}{t}\sum_{\psi\in  R(\grd,t)}\Pr[(\psi\preccurlyeq\Phi)\wedge(\topt\text{ covers } f)]\label{eq:score-LB-1}\\
        &=&\frac{L \beta}{t}\Pr\left[(\grd\text{ doesn't terminate by time }t  )\wedge(\topt\text{ covers } f)\right] \label{eq:score-LB-2}\\
        &\ge& \frac{L \beta }{t}\left( \Pr[\grd\text{ doesn't terminate by time }t  ] \,-\,  \Pr[\topt\text{ does not cover } f]\right)\label{eq:score-LB-3}\\
        &=&\frac{L \beta}{t}\cdot \left(  a(t)-o(t/( \beta L))\right) \label{eq:score-LB-4}
\end{eqnarray}

Inequality~\eqref{eq:score-LB-1} is by \eqref{eq:LB-score-psi}. The equality in \eqref{eq:score-LB-2} uses the definition of 
 $R(\grd,t)$.  Inequality~\eqref{eq:score-LB-3} is by a union bound. Equation \eqref{eq:score-LB-4} is by definition of the functions $a(\cdot)$ and $o(\cdot)$.

\paragraph{Proof of \eqref{eq:LB-score-psi}} Henceforth, fix any partial realization $\psi\in R(\grd,t)$. Our proof  relies on the following quantity:
\begin{align}
    Z:=\E_{\Phi}\left[\mathds{1}(\psi\preccurlyeq\Phi)\cdot \frac{f(\psi\cup P_{\infty})-f(\psi)}{Q-f(\psi)}\right]
\end{align}

In other words, this is the expected increase in policy $\topt$'s function value (relative to the ``remaining'' target   $Q-f(\psi)$) when restricted to (full) realizations $\Phi$ that agree with partial realization $\psi$. 

For any partial realization $\psi'$ such that $\psi \preccurlyeq\psi'$ and item $e\not\in \dom(\psi')$, let $X_{e,\psi,\psi'}$ denote the indicator r.v. that policy $\topt$ selects item $e$ at some point when its observed realizations are precisely $(\psi'\setminus \psi)\cup \chi$ where $\chi\sse \psi$. That is, $X_{e,\psi,\psi'}=1$ if policy $\topt$ selects $e$ at a point where (i) all items in $\dom(\psi'\setminus \psi)$ have been selected and their realization is $\psi'\setminus \psi$, (ii) no item in $E\setminus \dom(\psi')$ has been selected, and (iii) if any item in $\dom(\psi)$ has been selected then its realization agrees with $\psi$. Note that conditioned on $\psi' \preccurlyeq\Phi$, 
$X_{e,\psi,\psi'}$ is a deterministic value: the realizations of items $\dom(\psi')$ are fixed by $\psi'$ and if any item in $E\setminus \dom(\psi')$ is  selected (before $e$) then $X_{e,\psi,\psi'}=0$ irrespective of its realization. 

We can write $Z$ as a sum of increments as follows:
    \begin{align}
        Z&=\frac{1}{Q-f(\psi)}\E_{\Phi}\left[\mathds{1}(\psi\preccurlyeq\Phi)\cdot \sum_{k\geq 1}[f(\psi\cup P_k)-f(\psi\cup P_{k-1})]\right]\nonumber\\
        &=\frac{1}{Q-f(\psi)}\sum_{\psi':\psi\preccurlyeq\psi'}\sum_{e\notin \dom(\psi')}\E_{\Phi}\left[\mathds{1}(\psi'\preccurlyeq\Phi)\cdot X_{e,\psi,\psi'}\cdot [f(\psi'\cup(e,\Phi_e))-f(\psi')]\right]\nonumber\\
        &=\frac{1}{Q-f(\psi)}\sum_{\psi':\psi\preccurlyeq\psi'}\sum_{e\notin \dom(\psi')}\Pr[\psi'\preccurlyeq\Phi \wedge X_{e,\psi,\psi'}=1 ]\cdot \E_{\Phi}[f(\psi'\cup(e,\Phi_e))-f(\psi')\,|\,(\psi'\preccurlyeq\Phi)\wedge(X_{e,\psi,\psi'}=1)]\nonumber\\
        &=\frac{1}{Q-f(\psi)}\sum_{\psi':\psi\preccurlyeq\psi'}\sum_{e\notin \dom(\psi')}\Pr[\psi'\preccurlyeq\Phi \wedge X_{e,\psi,\psi'}=1 ]\cdot \E_{\Phi}[f(\psi'\cup(e,\Phi_e))-f(\psi')\,|\, \psi'\preccurlyeq\Phi ]\label{eq:score-LB-condn}\\
        &=\frac{1}{Q-f(\psi)}\sum_{\psi':\psi\preccurlyeq\psi'}\sum_{e\notin \dom(\psi')}\Pr[\psi'\preccurlyeq\Phi \wedge X_{e,\psi,\psi'}=1 ]\cdot \Delta(e|\psi')\nonumber
        \end{align}
        
        \begin{align}
        &\leq \frac{1}{Q-f(\psi)}\sum_{\psi':\psi\preccurlyeq\psi'}\sum_{e\notin \dom(\psi')}\Pr[\psi'\preccurlyeq\Phi \wedge X_{e,\psi,\psi'}=1 ]\cdot \Delta(e|\psi)\label{eq:score-LB-submod}\\
        &=\sum_{\psi':\psi\preccurlyeq\psi'}\sum_{e\notin \dom(\psi')}\Pr[\psi'\preccurlyeq\Phi \wedge X_{e,\psi,\psi'}=1 ]\cdot c_e\cdot \frac{\Delta(e|\psi)}{c_e(Q-f(\psi))}\nonumber\\
        &\leq \sum_{\psi':\psi\preccurlyeq\psi'}\sum_{e\notin \dom(\psi')}\Pr[\psi'\preccurlyeq\Phi \wedge X_{e,\psi,\psi'}=1 ]\cdot c_e\cdot  score(t,\psi)\label{eq:score-LB-greedy}\\
        &=  score(t,\psi)\sum_{\psi':\psi\preccurlyeq\psi'}\sum_{e\notin \dom(\psi')}c_e\cdot \Pr[\psi'\preccurlyeq\Phi \wedge X_{e,\psi,\psi'}=1 ]\nonumber\\
        &=  score(t,\psi)\cdot \sum_{e\in E\setminus \dom(\psi)} c_e\cdot \sum_{\substack{\psi':\psi\preccurlyeq\psi' \\  \dom(\psi')\not\ni e}}  \Pr[\psi'\preccurlyeq\Phi \wedge X_{e,\psi,\psi'}=1 ]\notag \\
        &=  score(t,\psi)\cdot \sum_{e\in E\setminus \dom(\psi)} c_e\cdot \Pr[\psi\preccurlyeq\Phi \wedge e\in I_\infty] \label{eq:score-LB-cost}\\
                &\le   score(t,\psi)\cdot \E_{\Phi}\left[\mathds{1}(\psi\preccurlyeq\Phi)\cdot \sum_{e\in I_{\infty}}c_e\right] \notag\\
        & \leq    score(t,\psi)\cdot (t/(L\beta)) \cdot \E_{\Phi}[\mathds{1}(\psi\preccurlyeq\Phi)]  \,\, =\,\,  score(t,\psi)\cdot (t/(L\beta))\cdot \Pr[\psi\preccurlyeq\Phi]   \label{eq:score-LB-Z} .
    \end{align}

    The equality  \eqref{eq:score-LB-condn} uses the fact that $X_{e,\psi,\psi'}$ is deterministic when conditioned on $\psi' \preccurlyeq\Phi$. Inequality \eqref{eq:score-LB-submod} is by adaptive submodularity. \eqref{eq:score-LB-greedy} is by the greedy selection criterion. The inequality in \eqref{eq:score-LB-Z} uses the fact that the total cost of $\topt$'s selections is always bounded above by $t/(L\beta)$. 
    Equation~\eqref{eq:score-LB-cost} uses the definition of $I_\infty$ (all selected items in $\topt$) and the following identity:
\[  \sum_{\substack{\psi':\psi\preccurlyeq\psi' \\  \dom(\psi')\not\ni e}}  \mathds{1}(\psi'\preccurlyeq\Phi)\cdot  X_{e,\psi,\psi'} = \mathds{1}(\psi\preccurlyeq\Phi \wedge e\in I_\infty) , \qquad \forall e\in E\setminus \dom(\psi).\]

To see this, condition on any full realization $\phi$. If $\psi \not \preccurlyeq\phi$ then both the left-hand-side ($LHS$) and right-hand-side ($RHS$) are $0$. If $\psi  \preccurlyeq\phi$ and $e$ is not selected by $\topt$ under $\phi$, then  again $LHS = RHS=0$. If $\psi  \preccurlyeq\phi$ and $e$ is  selected by $\topt$ under $\phi$, then $RHS=1$ and $LHS$ is the sum of $X_{e,\psi,\psi'}$ over $\psi'$ such that $\psi\preccurlyeq \psi'\preccurlyeq \phi$ and $e\not\in \dom(\psi')$. In this case, $X_{e,\psi,\psi'}=1$ for exactly one such  partial realization $\psi'$, namely   $\psi'=\psi \cup \kappa$ where $\kappa\preccurlyeq \phi$ is the partial realization immediately before $e$ is selected. So, $LHS=RHS$ in all cases. 

Note that whenever $\topt$ covers $f$, we have $f(P_\infty)=Q$. Combined with the monotone property of $f$, we have $f(\psi\cup P_\infty)=Q$ whenever $\topt$ covers $f$. So, we have:
$$Z\ge   \Pr[(\psi\preccurlyeq\Phi)\wedge(\topt\text{ covers } f)].$$

Combining the above inequality with \eqref{eq:score-LB-Z} finishes the proof of \eqref{eq:LB-score-psi}.
\end{proof}

\subsection{Wrapping Up}
We are now ready to complete the proof of Theorem~\ref{Theorem: THE_thmp}. 
Using Lemmas \ref{Lemma: upper_bd_lemma} and  \ref{Lemma: low_bd_lemma}, we get:
\begin{equation*}
    a(i) L\,\,\geq\,\, G_i \,\,\geq\,\, L\beta \int_{i}^{\infty}\frac{a(t)-o(t/(L \beta ))}{t}dt , \qquad \forall i\ge 0. %\label{eq:G-LUB}
\end{equation*}
Multiplying this inequality by $i^{p-1}$ and integrating over all $i\ge 0$, we obtain:
\begin{equation}
    \int_{0}^{\infty}i^{p-1}a(i)   di \,\,\geq\,\,    \beta \cdot \int_{i=0}^{\infty}i^{p-1}  \int_{t=i}^{\infty}\frac{a(t)-o(t/(L \beta ))}{t} \,dt\, di.  \label{eq:G-LUBp}
\end{equation}
Now, using  \eqref{eq:cp-grd}, the $p^{th}$ moment of the greedy cost is
 \begin{align}
      \cp(\grd) &= p \cdot     \int_{0}^{\infty}i^{p-1}a(i)  di \,\,\ge\,\, \beta p\cdot   \int_{i= 0}^{\infty}i^{p-1}\int_{t=i}^{\infty}\frac{a(t)-o(t/(L \beta))}{t} \,dt\, di \label{eq:p2.12}\\
      & = \beta p\cdot \int_{t=0}^{\infty}\frac{a(t)-o(t/(L \beta))}{t} \int_{i=0}^{t} i^{p-1}\,di \,dt  \,\,   =\,\,  \beta\cdot \int_{ 0}^{\infty}\frac{a(t)-o(t/L\beta)}{t} \cdot t^p dt   \label{eq:fb2}\\
      & = \beta\cdot \int_{ 0}^{\infty} t^{p-1}\cdot a(t) dt \, -\, \beta\cdot \int_{0}^{\infty}t^{p-1} \cdot o(t /L\beta)  dt  \nonumber \\
      & = \frac{\beta}p \cdot \cp(\grd)  \, -\, \beta\cdot \int_{0}^{\infty}t^{p-1} \cdot o(t /L\beta)  dt   \label{eq:a} \\
      &= \frac{\beta}p \cdot \cp(\grd)  \, -\,  \beta (L\beta)^{p} \cdot  \int_{y=0}^{\infty} y^{p-1}o(y) dy \,\,=\,\, \frac{\beta}p \cdot \cp(\grd)  \, -\,  \frac{\beta}p \cdot (L\beta)^{p} \cdot \cp(\opt) \label{eq:o}.
\end{align}
The inequality in  \eqref{eq:p2.12} is by \eqref{eq:G-LUBp}. The first equality in \eqref{eq:fb2} is by interchanging the order of the integrals (by Fubini's theorem).  Equality \eqref{eq:a} uses \eqref{eq:cp-grd} for $\cp(\grd)$. The first equality in  \eqref{eq:o} is by a change of variable $y=\frac{t}{L\beta}$ in the integral, and the last equality uses \eqref{eq:cp-opt} for $\cp(\opt)$.
  
 It now follows that 
$$\left(\frac{\beta}{p} -1\right) \cdot \cp(\grd) \quad \le \quad \frac{\beta}p \cdot (L\beta)^{p} \cdot   \cp(\opt).$$
In order to minimize the approximation ratio, we choose $\beta=p+1$ (note that this is only used in the analysis), which implies:
$$\cp(\grd) \quad \le \quad (p+1)^{p+1}\cdot L^p \cdot \cp(\opt).$$

 This completes the proof of Theorem~\ref{Theorem: THE_thmp}. Setting $p=1$, we also obtain   Theorem~\ref{Theorem: THE_thm}.

\section{Applications}\label{sec:appln}
Here, we provide some concrete applications of our framework. These applications were already discussed in \cite{GK17}, but as noted in Section~\ref{sec:Definition}, the function definition in \asc is slightly more restrictive than the framework in \cite{GK17}.

\paragraph{Stochastic Submodular Cover.} In this problem, there are $n$ stochastic items (for example, corresponding to sensors). Each item $e$ can be in one of many  ``states'', and this state is observed only after selecting item $e$. E.g.,  the state of a sensor indicates the extent to which it is working. The states of different items are independent. There is a utility function $\hat{f}: 2^{E\times \Omega}\rightarrow \R_{\ge 0}$, where  $E$ is the set of items and $\Omega$ the set of states. It is assumed that $\hat{f}$ is monotone and submodular. For e.g., $\hat{f}$ quantifies the information gained from a set of sensors having arbitrary states. Each item $e$ is also associated with a cost $c_e$. Given a quota $Q$, the goal is to select items sequentially to achieve utility at least $Q$, at the minimum expected cost.  We assume that the quota $Q$ can always be achieved by selecting adequately many items, i.e., $\hat{f}(\{(e,\phi_e):e\in E\})\ge Q$ for all possible states $\{\phi_e\in \Omega\}_{e\in E}$ for the items. This is a special case of \asc, where the items $E$ and states (outcomes) $\Omega$ remain the same. We define a new utility function $f(\psi) = \min\left\{ \hat{f}(\psi) , Q\right\}$ for all $\psi\sse E\times \Omega$.   Note that $Q$ is the maximal value of function $f$ and this value is achieved under every possible (full) realization. Moreover, $f$ is also monotone and submodular. Clearly, the monotonicity property (Definition~\ref{Definition:mon}) holds. The adaptive-submodularity property also holds because the items are independent. Indeed, for any partial realizations $\psi\preccurlyeq \psi'$ and $e\in E\setminus\dom(\psi')$,
\begin{eqnarray*}
\Delta(e|\psi) &=& \sum_{\omega\in \Omega}\Pr[\Phi_e=\omega|\psi\preccurlyeq\Phi]\cdot \left( f(\psi\cup (e,\omega) )-f(\psi)\right) \\
&=& \sum_{\omega\in \Omega}\Pr[\Phi_e=\omega]\cdot \left( f(\psi\cup (e,\omega) )-f(\psi)\right)  \\
&\ge &\sum_{\omega\in \Omega}\Pr[\Phi_e=\omega]\cdot \left( f(\psi'\cup (e,\omega) )-f(\psi')\right)  \\
&= & \sum_{\omega\in \Omega}\Pr[\Phi_e=\omega|\psi'\preccurlyeq\Phi]\cdot \left( f(\psi'\cup (e,\omega) )-f(\psi')\right) = \Delta(e|\psi'). 
\end{eqnarray*}

In particular, when $\hat{f}$ is integer-valued, Theorem~\ref{Theorem: THE_thm} implies a $4(1+\ln Q)$-approximation algorithm. We note that \cite{HellersteinKP21} obtained a $(1+\ln Q)$-approximation ratio, using a different analysis. The latter bound is the best possible, including the constant factor, as the problem generalizes set cover. However, our approach is more versatile and also provides a $(p+1)^{p+1}\cdot (1+\ln Q)^p$ approximation  bound for the $p^{th}$ moment objective. Our result is the first approximation algorithm for higher moments ($p>1$).

\paragraph{Adaptive Viral Marketing.} This problem is defined on a directed graph $G=(V,A)$ representing a social network~\cite{KempeKT15}. Each node $v\in V$ represents a user. Each arc $(u,v)\in A$ is associated with a random variable $X_{uv}\in \{0,1\}$. The r.v. $X_{uv}=1$  if $u$ will influence $v$ (assuming $u$ itself is influenced); we also say that arc $(u,v)$ is {\em active} in this case. The r.v.s $X_{uv}$ are independent, and we are given the means $\E[X_{uv}]=p_{uv}$ for all $(u,v)\in A$. When a node $u$ is activated/influenced, all arcs $(u,v)$ out of $u$ are observed and if $X_{uv}=1$ then $v$ is also activated. This process then continues on $u$'s neighbors to their neighbors and so on, until no new node is activated. We consider the ``full feedback'' model, where after activating a node $w$, we observe the $X_{uv}$ r.v.s on all arcs $(u,v)$ such that $u$ is reachable from $w$ via a path of active arcs. Further, each node $v$ has a cost $c_v$ corresponding to activating node $v$ {\em directly}, e.g. by providing some promotional offer. Note that there is no cost incurred on $v$ if  it is activated (indirectly) due to a  neighbor $u$ with $X_{uv}=1$.  Given a quota $Q$, the goal is to activate at least $Q$ nodes at the minimum expected cost. 

To model this as \asc, the items $E=V$ are all nodes in $G$. We add  self-loops $A_o=\{(v,v) : v\in V\}$ that represent whether a node is activated directly. So, the new set of  arcs is $A'=A\cup A_o$. The outcome $\Phi_w$ of any node $w\in V$ is represented by a function $\phi_w:A'\rightarrow \{0,1,?\}$ where $\phi_w((w,w))=1$, $\phi_w((v,v))=0$ for all $v\in V\setminus w$, and for any $(u,v)\in A$:
\begin{itemize}
    \item $\phi_w(u,v)=1$ if there is a $w-u$ path of active arcs and $X_{uv}=1$ (i.e., $(u,v)$ is active).

\item $\phi_w(u,v)=0$ if there is a $w-u$ path of active arcs and $X_{uv}=0$ (i.e., $(u,v)$ is not active).

\item $\phi_w(u,v)=?$ if there is no $w-u$ path of active arcs (so, the status of  $(u,v)$ is unknown).

\end{itemize}
Let $\Omega$ denote the collection of all such functions: this represents the outcome space. Note that $\Phi_w$ depends on the entire network (and not just node $w$). So, the r.v.s $\{\Phi_w\}_{w\in V}$ may be highly correlated.
Observe that $\Phi_w$ is exactly the feedback obtained when node $w$ is activated directly (by incurring cost $c_w$) at any point in a policy. Define function 
$\bar{f}:2^{E\times \Omega}\rightarrow \R_{\ge 0}$ as:
\begin{equation}
    \label{eq:f-influence} \bar{f}(\psi) = \sum_{v\in V} \min\left\{ \sum_{u:(u,v)\in A'} |\{ w\in \dom(\psi) : \psi_w(u,v) =1\}| \, ,\, 1\right\}.\end{equation}
$\bar{f}$ is a sum of set-coverage functions, which is monotone and submodular. 
Then, utility function $f:2^{E\times \Omega}\rightarrow \R_{\ge 0}$ is $f(\psi)=\min\{\bar{f}(\psi), Q\}$. Function $f$ is clearly monotone (Definition~\ref{Definition:mon}). The adaptive-submodularity property also holds: see Theorem~19 in \cite{GK17}.  

Hence, Theorem~\ref{Theorem: THE_thm} implies a $4(1+\ln Q)$-approximation algorithm for adaptive viral marketing. This is an improvement over previous  approximation ratios of $(1+\ln Q)^2$~\cite{GK17} and   $(1+\ln (nQc_{max}))$~\cite{EKM21}, where $n=|V|$ and $c_{max}$ is the maximum cost. We also obtain the first approximation algorithm  for the $p^{th}$ moment objective. 

\paragraph{Optimal Decision Tree (uniform prior).} In this problem, there are $m$ hypotheses $H$ and $n$ binary tests $E$. Each test $e\in E$  costs $c_e$, and has a positive outcome on some subset $T_e\sse H$ of hypotheses (its outcome is negative on the other hypotheses).\footnote{Our results also extend to the case of multiway tests with non-binary outcomes.} An unknown  hypothesis $h^*$ is drawn from $H$ uniformly at random. The goal is to identify $h^*$ by sequentially performing tests, at minimum expected cost.  This is a special case of \asc, where the items correspond to tests $E$ and the outcome space $\Omega=\{+,-\}$. The outcome $\Phi_e$ for any item $e$ is the test outcome under the (unknown) hypothesis $h^*$. \label{subsec:odt}
For any test $e\in E$, define subsets $S_{e,+} = H\setminus T_e$ and $S_{e,-} =T_e$, corresponding to the hypotheses that can be eliminated when we observe a positive or negative outcome on $e$. The utility function is
$$f(\psi) = \frac1{|H|} \cdot \big|\bigcup_{e\in \dom(\psi)} S_{e,\psi_e}\big|.$$
The quota $Q=1-\frac1{|H|}$. Achieving value $Q$ means that $|H|-1$ hypotheses have been eliminated, which implies that $h^*$ is identified.  
The function $f$ is again monotone and submodular. The monotonicity property (Definition~\ref{Definition:mon}) clearly holds. Moreover, using the fact that $h^*$ has a {\em uniform distribution}, it is known  that $f$ is adaptive-submodular: see Lemma~23 in \cite{GK17}.  

So, Theorem~\ref{Theorem: THE_thm} implies a $4(1+\ln (|H|-1))$-approximation algorithm for this problem; we use $Q$ as above and $\eta=\frac{1}{|H|}$. The previous-best bounds for this problem were $(1+\ln (|H|-1))^2$~\cite{GK17}, $12\cdot\ln|H|$~\cite{GuilloryB09} and   $(1+\ln (n|H|c_{max}))$~\cite{EKM21}. Again, our result is the first approximation algorithm for this problem under $p^{th}$ moment objectives.

We note that \cite{GK17} also obtained a $\left(\ln \frac{1}{p_{min}}\right)^2$-approximation for the optimal decision tree problem with arbitrary priors (where the distribution of $h^*$ is not uniform); here $p_{min}\le \frac1{|H|}$ is the minimum probability of any hypothesis. This uses a different utility function that falls outside our \asc framework (as our definition of  function  $f$ is more restrictive). Moreover, there are other approaches~\cite{GNR17,NavidiKN20} that provide a better $O(\ln |H|)$-approximation bound even for the problem with arbitrary priors. 

 \def\masc{\ensuremath{{\sf MASC}}\xspace}
\def\cov{\ensuremath{C}\xspace}
\def\ii{r}
\def\cpm{\widetilde{c}_p}

\section{Adaptive Submodular Cover  with Multiple  Functions} Here, we extend \asc to the setting of covering multiple adaptive-submodular functions. In the multiple adaptive-submodular cover (\masc) problem,  there is a set $E$  of  items and outcome space $\Omega$ as before. Each item $e\in E$ has a cost $c_e$; we will view this  cost as the item's {\em processing time}. Now, there are $k$ different utility functions $f_{\ii}:2^{E\times\Omega}\to\R_{\geq 0}$ for $\ii\in[k]$. We assume that each of these functions satisfies the monotonicity, coverability and adaptive-submodularity properties. We also assume, without loss of generality (by scaling), that the maximal value of each function $\{f_{\ii}\}_{\ii=1}^k$ is $Q$. As for the basic \asc problem, a solution to \masc  corresponds to a policy $\pi:2^{E\times \Omega}\rightarrow E$, that maps partial realizations to the next item to select. Given any policy $\pi$, the {\em cover time}  of function $f_{\ii}$ is defined as:
$$C_{\ii}(\pi) := \mbox{the earliest time $t$ such that } f_{\ii}(\Psi(\pi,t))=Q.$$
Recall that $\Psi(\pi,t)\sse E\times \Omega$  is the partial realization that has been observed by time $t$ in policy $\pi$. Note that the cover time is a random quantity.  
The expected cost objective in \masc  is  the expected total cover time of all functions, i.e., $\sum_{\ii=1}^k \E[\cov_{\ii}(\pi)]$. More generally, for any $p\ge 1$, the $p^{th}$  moment objective of policy $\pi$ is 
$\cpm(\pi) := \sum_{\ii=1}^k \E\left[\cov_{\ii}(\pi)^p\right]$. When we have just $k=1$ function, the  \masc problem reduces to \asc. 

\paragraph{Remark:} One might also consider an alternative multiple-function formulation where we are interested in the expected {\em maximum} cover time of the functions. This formulation can be directly solved as an instance of \asc where we use the single adaptive-submodular function  
   $g=\sum_{\ii=1}^k f_{\ii}$ with maximal value $Q'=kQ$.

\medskip

We extend the greedy policy for \asc to \masc, as described in Algorithm~\ref{alg:algo_multiple}. 
For each $\ii\in [k]$ and item $e\in E$, we use $\Delta_{\ii}(e|\psi)$ to denote the marginal benefit of $e$  under function $f_{\ii}$.  Notice that the greedy selection criterion here involves a sum of terms corresponding to each un-covered function. A similar greedy rule was used earlier in the (deterministic) submodular function ranking problem \cite{AzarG11} and in the independent special case of \masc in \cite{INZ12}.

\begin{algorithm}
\caption{Adaptive Greedy Policy $\grd$ for  \masc.\label{alg:algo_multiple}}
\begin{algorithmic}[1]
\State selected items $A \gets \emptyset$, observed realizations $\psi \gets \emptyset$
\While{there exists function $f_{\ii}$ with $f_{\ii}(\psi)<Q$}
\State select item
$$e^{*} = \argmax_{e\in E\setminus A}\,\, \frac{1}{c_e}\cdot\sum_{\ii\in [k]:f_{\ii}(\psi)<Q}\frac{\Delta_{\ii}(e|\psi)}{Q-f_{\ii}(\psi)},$$
and observe $\Phi_{e^*}$ 
\State add $e^*$ to the selected items, i.e., $A \gets A\cup\{e^{*}\}$ 
\State update $\psi \gets \psi\cup\{(e^{*}, \Phi_{e^{*}})\}$
\EndWhile
 \end{algorithmic}
\end{algorithm}

\begin{thm}\label{Theorem: multi_thmp}
Consider any instance of  adaptive-submodular cover  with $k$ utility functions, where each function $f_{\ii}:2^{E\times\Omega}\to\R_{\ge 0}$ is monotone, coverable and adaptive-submodular w.r.t. the same probability distribution $q(\cdot)$. Suppose that there is some value $\eta>0$ such that $f_{\ii}(\psi)>Q-\eta$ implies $f_{\ii}(\psi)=Q$ for all partial realizations $\psi\sse E\times \Omega$ and $\ii\in [k]$. Then, for every $p\ge 1$, the $p^{th}$ moment of the cost of the greedy policy is
\[\cpm(\grd)\leq (p+1)^{p+1}\cdot (1+\ln (Q/\eta))\cdot  \cpm(\opt_p),\]
where $\opt_p$ is the optimal \masc policy for the $p^{th}$ moment objective.\end{thm}
In particular, setting $p=1$, we obtain a $4(1+\ln (Q/\eta))$ approximation algorithm for \masc with the expected cost objective.
The proof of Theorem~\ref{Theorem: multi_thmp} is a natural extension  of   Theorem \ref{Theorem: THE_thmp} for \asc. We use the same notations and definitions if not mentioned explicitly.

\begin{defn}[\masc non-completion probabilities] For any time $t\ge 0$ and $\ii\in [k]$, let 
    \begin{equation*}
    \begin{aligned}
    \oo_{\ii}(t)\,:&= \, \Pr[\opt \mbox{ does not cover } f_{\ii}\mbox{ by time }t]  = \Pr[\cov_{\ii}(\opt) >  t] =  \Pr\left[f_{\ii}(\Psi(\opt , t)) < Q\right].\\
    \aa_{\ii}(t)\,:&= \, \Pr[\grd \mbox{ does not cover } f_{\ii}\mbox{ by time } t]  = \Pr[\cov_{\ii}(\grd) >  t] =  \Pr\left[f_{\ii}(\Psi(\grd , t)) < Q\right].\\
    \end{aligned}
    \end{equation*}
Also, define  for any $t\ge 0$,  $\oo(t) :=\sum_{\ii=1}^{k} \oo_{\ii}(t)$ and $        \aa(t) :=\sum_{\ii=1}^{k}\aa_{\ii}(t)$.
\label{defn:masc}
    \end{defn}

For each $\ii\in [k]$, the functions   $\oo_\ii(t)$ and $\aa_\ii(t)$ are non-increasing functions of $t$; moreover, $\oo_\ii(0)=\aa_\ii(0)=1$ and  $\oo_\ii(t)=\aa_\ii(t)=0$ for all $t\ge \sum_{e\in E} c_e$ (this is because any policy would have selected all items by this time). So, the functions $\oo_\ii(t)$ and $\aa_\ii(t)$ 
 share the same properties as $o(t)$ and $a(t)$ for \asc.  As in \eqref{eq:cp-grd} and \eqref{eq:cp-opt}, we can express the $p^{th}$ moment objective in \masc as: 
\begin{equation*}
\cpm(\opt)=    p \int_{0}^{\infty}  t^{p-1} \cdot   \oo(t) dt, \quad \mbox{and} \quad \cpm(\grd) = p \int_{0}^{\infty}  t^{p-1} \cdot   \aa(t) dt.
\end{equation*}

\begin{defn}[\masc greedy score]
For any $t\ge 0$ and partial realization $\psi$ observed at time $t$, we define
    \begin{equation*}
\sscore(t,\psi) := \frac1{c_e}   \sum_{\ii\in [k]:f_{\ii}(\psi)<Q}\,\frac{\Delta_{\ii}(e|\psi)}{ Q-f_{\ii}(\psi)}\quad 
        \substack{\text{where $e$ is the item being selected in $\grd$ at time $t$ } \\
        \text{ and $\psi$ was observed just  before selecting } e}, 
    \end{equation*}    
and  $\sscore(t,\psi):=0$ if no item is being selected  in $\grd$ at time $t$ when $\psi$ was observed.
\end{defn}

\begin{defn}[\masc potential gain]
For any time   $i\ge 0$, its potential  gain  is    the expected total score accumulated after time $i$,  
    \[\Tilde{G}_{i} := \int_{i}^{\infty } \E \left[\sscore(t,\Psi_t) \right]dt \]
\end{defn}

The next two lemmas  lower and upper bound the potential gain. 
\begin{lem}\label{Lemma: upper_bd_lemma_multi}
For any $i\ge 0$
\[\Tilde{G}_i\leq \sum_{\ii\in [k]}L\cdot \aa_{\ii}(t)=L\cdot \aa(t).\]
\end{lem}
 Lemma \ref{Lemma: upper_bd_lemma_multi} follows by  applying Lemma \ref{Lemma: upper_bd_lemma} to each $f_{\ii}$ and adding over $\ii\in[k]$.

\begin{lem}\label{Lemma: low_bd_lemma_multi}
For any $t\ge 0$ , 
\[\E[\sscore(t, \Psi_t)]\geq \sum_{\ii\in [k]}\frac{\aa_{\ii}(t)-\oo_\ii(t/(L\beta))}{t/(L\beta)}=\frac{\aa(t)-\oo(t/(L\beta))}{t/(L\beta)}.\]
Hence, $\Tilde{G}_i \ge L\beta  \int_{i}^{\infty} \left( \frac{\aa(t)-\oo(t/(L\beta)))}{t} \right) dt $ for each time $i\ge 0$.
\end{lem}
\begin{proof}[Proof Outline]
    We   replicate all the steps in Lemma \ref{Lemma: low_bd_lemma}, except that we redefine $Z$ to be
\[\Tilde{Z}:=\E_{\Phi}\left[\mathds{1}(\psi\preccurlyeq\Phi)\cdot \sum_{\ii\in[k]:f_r(\psi)<Q}\frac{f_\ii(\psi\cup P_{\infty})-f_\ii(\psi)}{Q-f_\ii(\psi)}\right],\]
which takes all $k$ functions into account. Here, $\psi$ is any partial realization observed at time $t$. We then obtain:
$$\sscore(t,\psi)\cdot \frac{t}{L \beta}\cdot \Pr[\psi\preccurlyeq\Phi] \ge \Tilde{Z} \ge \sum_{\ii\in[k]:f_\ii(\psi)<Q}  \Pr[(\psi\preccurlyeq\Phi)\wedge(\topt\text{ covers } f_{\ii})].$$ 

Let $R(\grd,t)$ denote {\em all} the possible partial realizations observed at time $t$ in policy $\grd$. Then,  
\begin{eqnarray*}
        \E[\sscore(t,\Psi_t)]& = & \sum_{\psi\in R(\grd,t)} \Pr[\psi \preccurlyeq \Phi ] \cdot \sscore(t,\psi) \\
        &\geq& \frac{L\beta}{t}\sum_{\psi\in R(\grd,t)}\quad  \sum_{\ii\in[k]:f_\ii(\psi)<Q}  \Pr[(\psi\preccurlyeq\Phi)\wedge(\topt\text{ covers } f_{\ii})]\\
        &=&\frac{L\beta}{t}\sum_{\ii\in [k]} \Pr\left[(\grd \mbox{ doesn't cover } f_{\ii}\mbox{ by time }t  )\wedge(\topt\text{ covers } f_{\ii})\right]  \\
        &\ge& \frac{L \beta}{t}\sum_{\ii\in [k]}\left( \Pr[\grd \mbox{ doesn't cover } f_{\ii}\mbox{ by time }t ]  -   \Pr[\topt\text{ doesn't cover } f_\ii]\right) \\
        &=&\frac{L \beta }{t}\sum_{\ii\in [k]} \left(\aa_\ii\left(t\right)-o_\ii(t/(L \beta))\right)  \geq L\beta \frac{\aa(t)-\oo(t/(L \beta))}{t}. 
\end{eqnarray*}
\end{proof}

Finally, we combine Lemma \ref{Lemma: upper_bd_lemma_multi} and Lemma \ref{Lemma: low_bd_lemma_multi},  as in Theorem \ref{Theorem: THE_thmp}, to obtain:
$$\cpm(\grd)\leq \frac{L^p\beta^{p+1}}{\beta-p}\cdot \cpm(\opt).$$
Setting $\beta=p+1$ to minimize the approximation ratio, we obtain Theorem~\ref{Theorem: multi_thmp}. 

We now list some applications of the \masc result.
\begin{itemize}
\item When items are deterministic, \masc reduces to the deterministic submodular ranking problem, for which an $O(\ln (Q/\eta))$ was obtained in \cite{AzarG11}. We note that the result in \cite{AzarG11} was only for unit costs, whereas our result holds for arbitrary costs. 
 This problem generalizes the min-sum set cover problem, which is NP-hard to approximate better than factor $4$ \cite{FLT04}. For min-sum set cover, the paramters $Q=\eta=1$: so Theorem \ref{Theorem: multi_thmp} implies a {\em tight} $4$-approximation algorithm for it.  
 \item When the outcomes are independent across items, \masc reduces to stochastic submodular cover with multiple functions, which was studied in \cite{INZ12}. We obtain an $4\cdot (1+\ln(Q/\eta))$ approximation ratio that improves the bound of  $56\cdot (1+\ln(Q/\eta))$ in \cite{INZ12} by a constant factor. Although \cite{INZ12} did not try to optimize the constant factor, their approach   seems unlikely to provide   such a small constant factor.  
    \item Consider the following generalization of adaptive viral marketing. Instead of a single quota on the number of influenced nodes, there are $k$ different quotas $Q_1\le Q_2 \le \cdots Q_k$. Now, we want a policy  such that the {\em average} expected cost for achieving these quotas is minimized.   Recall the function $\bar{f}$ defined in \eqref{eq:f-influence} for the single-quota problem. Then, corresponding to the different quotas,   define  functions $f_\ii(\psi) = \frac1{Q_\ii}\cdot \min\{\bar{f}(\psi) , Q_\ii\}$ for each $\ii\in [k]$. Each of these functions is monotone, adaptive-submodular and has maximal value $Q=1$. The parameter $\eta=1/ Q_k$, so we obtain a $4(1+\ln Q_k)$-approximation algorithm.
 \end{itemize}

\section{Computational Results}
In this section, we empirically evaluate the performance of the greedy policy on a real-world dataset of Optimal Decision Tree (ODT), which is one of the applications of \asc (see Section~\ref{subsec:odt}). This dataset has been used in a number of previous papers on ODT~\cite{BBS12,bhavnani2007network,NavidiKN20}. We also provide a lower-bound for the $p^{th}$ moment objective in ODT, and compare the greedy policy's performance to this bound. 

Recall that in ODT, there are $m$ hypotheses (with uniform probabilities) and $n$ binary tests. We assume that all tests have unit costs. The goal is to identify the realized hypothesis by performing sequential tests. The objective is to minimize the $p^{th}$ moment of the testing cost. The case $p=1$ is the usual expected cost objective, which has been analyzed before. We also consider higher moments with $p=2,3$, for which our result provides the first approximation ratio.   

\paragraph{Information theoretic lower bound for ODT.} It is well-known that the expected testing cost for any instance of ODT (with uniform probability and cost) is at least $\log_2(m)$. This follows from the entropy lower bound for binary coding. In fact, we can get a stronger lower bound using Huffman coding~\cite{huffman1952}, which  corresponds to  the ODT instance with $m$ hypotheses and all possible  binary tests. Viewed  this way, any  policy for this ODT instance is a binary tree $T$ with $m$ leaves $L(T)$:  the $p^{th}$ moment objective is then $\frac1m\sum_{i\in L(T)} d(i)^p$ where $d(i)$ is the depth of leaf $i$. 
Below, we show that the Huffman coding construction   also minimizes the $p^{th}$ moment: so it provides a lower bound for ODT with  $p^{th}$ moment objectives.

We define the \emph{Huffman  tree} with $m$ leaves to be a binary tree with all non-leaf nodes having two children, and exactly   $2^{\lceil{\log_2(m) \rceil}} -m$ leaves at depth $\lfloor{\log_2(m)\rfloor}$ and   $2m - 2^{\lceil{\log_2(m) \rceil}}$ leaves at depth $\lceil{\log_2(m)\rceil}$. When $m$ is a power of two, the Huffman  tree corresponds to the ``complete balanced binary tree'' with all  $m$ leaves at depth $\log_2(m)$.
The following lemma (proved in Appendix~\ref{app}) summarizes this lower bound.
 \begin{lem}
\label{lem:logp}
Among all binary trees with $m$ leaves, the Huffman tree minimizes the   sum of $p^{th}$ powers of the leaf-depths, for all $p\ge 1$. Hence, the optimal $p^{th}$ moment for any ODT instance with $m$ hypotheses and unit-cost tests is at least:
    $$ \frac1m\cdot \left(2^{\lceil{\log_2(m) \rceil}} -m \right)\cdot \lfloor{\log(m)\rfloor}^p\, +\, \frac2m\cdot  \left(2m - 2^{\lceil{\log_2(m) \rceil}}\right)\cdot \lceil{ \log(m) \rceil}^p.$$ 
\end{lem}

%%%%%%%%%%%%%%

\paragraph{Instances.} We use the WISER data (\href{http://wiser.nlm.nih.gov/}{http://
wiser.nlm.nih.gov/}), which   lists 415 toxins and 79 symptoms. Each toxin is associated with a set of observed symptoms that are reported in the data set. This corresponds to an ODT instance where the goal is to identify the toxin by testing for symptoms. For some toxin/symptom combinations it is unknown whether exposure to a toxin definitively exhibits a certain symptom. For such cases, we fill in randomly either a positive or negative outcome for the symptom, and generate 5 such variations, labelled as datasets $A_1$ to $A_5$. 
Given that we fill in the unknown symptoms randomly, some  hypotheses (i.e., toxins) are equivalent (i.e., have  identical set of symptoms); so   we remove all such  duplicates. This is why the number of hypotheses $m$ varies in the instances.

Our results show that the greedy policy is optimal (or extremely close to the optimum) on all the datasets  $A_1$ to $A_5$. 
This indicates that the WISER data set is highly structured with many balanced tests, and that the greedy policy is able to exploit this structure. In order to test the performance on less-structured instances, we also created 5 more instances (labeled $B_1$ to $B_5$) by restricting to a random subset of  15 (out of  79) tests. Again, we remove all duplicate hypotheses.

\paragraph{Results.} In Table~\ref{tab:greedybounds}, we report the performance of the greedy policy on each of the $10$ instances described above.  For each instance, we report the number $m$ of hypotheses and the empirical approximation ratio  (greedy objective divided by the lower-bound from Lemma~\ref{lem:logp}) for $p=1,2,3$. For the expectation objective ($p=1$), we   also report the objective value of the greedy policy and   the two  lower bounds:  entropy-based (which was also used in prior work) and the Huffman bound (Lemma~\ref{lem:logp}). 
We note that  the Huffman bound is indeed better than  the entropy-based lower-bound. For instances $A_1$ to $A_5$, the greedy policy's cost matches  the lower bound, resulting in an empirical approximation ratio of $1$. For the modified instances $B_1$ to $B_5$, our algorithm's performance is still very good: the empirical approximation for $p=1,2,3$ is at most $1.04, 1.09$ and $1.18$  (respectively).

\begin{center}
\begin{table} 
\begin{tabular}{ | c | c| c | c|c| c|c|c|} 
\hline
  \makecell{Instance} & $m$ & \makecell{Entropy \\ bound } & \makecell{Revised \\bound} & \makecell{Sum of \\ costs} & \makecell{Approx. \\factor $p=1$ }& \makecell{Approx. \\ factor $p=2$} & \makecell{Approx. \\factor $p=3$ }\\ 
  \hline
  $A_1$ & 405  & 3508.02 & 3538 & 3538 & 1.0000 & 1.0000 & 1.0000 \\ 
  \hline
  $A_2$ & 395 & 3407.16 & 3438 & 3438 & 1.0000  & 1.0000 &  1.0000\\ 
  \hline
  $A_3$ & 399 & 3447.46 & 3478 & 3479 & 1.0000 & 1.0001 &  1.0001\\ 
  \hline
  $A_4$ & 399 & 3447.46 & 3478 & 3478 & 1.0000 &  1.0000 & 1.0000 \\ 
  \hline
  $A_5$ & 395 & 3407.16 & 3438 & 3439 & 1.0003 & 1.0007 & 1.0012
 \\ 
  \hline
    $B_1$ & 207 & 1592.55 & 1607 & 1666 & 1.0367 & 1.0941
 & 1.1784 \\ 
  \hline
 $B_2$ & 248 & 1972.64 & 1976 & 2019 & 1.0218 & 1.0525 & 1.0931\\ 
  \hline
$B_3$ & 249 & 1982.04 & 1985 & 2025 & 1.0202 & 1.0488 & 1.0868\\ 
\hline
$B_{4}$ & 266 & 2142.71 & 2148 & 2211 & 1.0293 & 1.0715 & 1.1284\\ 
\hline
$B_{5}$ & 274 & 2218.86 & 2228 & 2269 & 1.0184 & 1.0440 & 1.0775\\ 
\hline
\end{tabular}
\caption{Computational results on WISER dataset. \label{tab:greedybounds}}
\end{table}
\end{center}

\section{Conclusions}
We studied  the adaptive-submodular cover problem with $p^{th}$ moment objectives, and showed that the natural greedy policy simultaneously  achieves a $(p+1)^{p+1}\cdot (\ln Q+1)^p$ approximation guarantee for all $p\ge 1$.  Even for the well-studied case of minimizing expected cost ($p=1$), our result provides  the first $O(\ln Q)$ approximation bound.  Our results also extend to the setting with multiple adaptive-submodular functions. 
 While our approximation ratios are the best possible up to constant factors, it would still be interesting to pin down the exact constant. For example, can we get a $\ln (Q)$ approximation ratio for minimizing the expected cost? 
\bibliographystyle{alpha}
\bibliography{references}

\appendix

    \section{Missing Proofs}\label{app}

    \subsection{Upper bounding the score}
    Below, we upper bound the total score for a single realization $\phi$. A similar fact and proof was used previously in \cite{HellersteinKP21,INZ12,AzarG11}. 
    \begin{lem}\label{Lemma:harmonic}
    For any (full) realization $\phi$,
    \[\int_{0}^{\infty}S(t,\phi) dt\leq L = \ln(Q/\eta) +1\]
    \end{lem}
        \begin{proof} 
    Under $\phi$, let $e_1, e_2, ..., e_k$ be the sequence of items selected by $\grd$, let $\psi_i$ be the partial realization just before selecting $e_i$, and define $f_i:=f(\psi_i)$. Note that $0\leq f_1\leq f_2\leq ...\leq f_{k+1}= Q$ by monotonicity and the assumption that $f$ is always covered by $\pi$. Moreover, we have $f_{k}\leq Q-\eta$ by the definition of $\eta$: otherwise we would have $f_k=Q$ and $\grd$ would terminate before selecting $e_k$.

    Define a function $g:[0,\infty)\to\R$ by

    \begin{equation}
g(x) := 
    \begin{cases}
        \frac{1}{Q-f_i}, & \text{for } x\in [f_i, f_{i+1}) \text{ and any } i=1,2,...,k-1,\\
        0, & \text{otherwise} .
    \end{cases}
\end{equation}
    
\begin{figure}[h]
    \centering   
\begin{tikzpicture}
\begin{axis}[
            xmin = 0, xmax = 10,
            ymin = 0, ymax = 0.5,
            xtick distance = 1,
            ytick distance = 1,
            grid = none,
            minor tick num = 0,
            major grid style = {lightgray},
            minor grid style = {lightgray!25},
            width = 0.65\textwidth,
            height = 0.35\textwidth,
            xlabel = {$x$},
            ylabel = {$g(x)$},
            legend cell align = {left},
        ]
\addplot[mark=*] coordinates {(0,0.1)};
\addplot[ultra thick][domain=0:2] {0.1};
\addplot[mark=*,fill=white] coordinates {(2,0.1)};
\addplot[mark=*] coordinates {(2,1/8)};
\addplot[ultra thick][domain=2:3.5] {1/8};
\addplot[mark=*,fill=white] coordinates {(3.5,1/8)};
\addplot[mark=*] coordinates {(3.5,1/6.5)};
\addplot[ultra thick][domain=3.5:7] {1/6.5};
\addplot[mark=*,fill=white] coordinates {(7,1/6.5)};
\addplot[mark=*] coordinates {(7,1/3)};
\addplot[ultra thick][domain=7:8.5] {1/3};
\addplot[mark=*,fill=white] coordinates {(8.5,1/3)};

\addplot[blue, domain=0:9.9] {1/(10-x)};
\node[blue, ultra thick] at (5,0.29) { $\overline{g}(x)=\frac{1}{Q-x}$};
\end{axis}
\end{tikzpicture}
    \caption{Example of $g(x)$, where $k=5,Q=10, \eta=1$ and $f_1=0, f_2=2, f_3=3.5, f_4=7, f_5=8.5$. }
    \label{fig:example_function_g}
\end{figure}
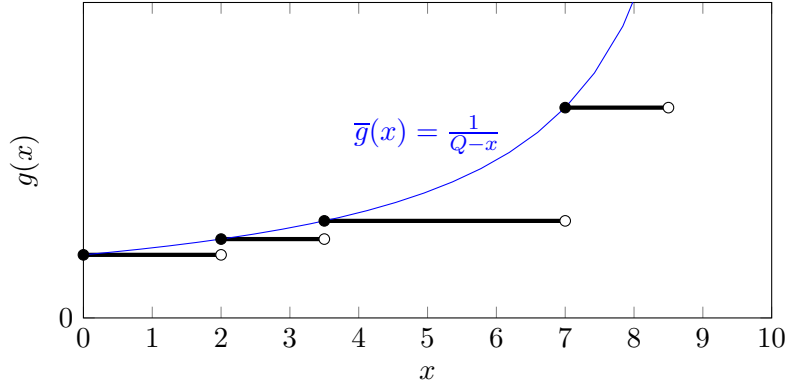

Note that $g(x)\le \frac{1}{Q-x}$ for all $0\le x\le Q-\eta$ (see Figure~\ref{fig:example_function_g} for an example). So,
 \[\int_{0}^{\infty}S(t,\phi) dt=\sum_{i=1}^{k}\frac{f_{i+1}-f_i}{Q-f_i} = \sum_{i=1}^{k-1}\frac{f_{i+1}-f_i}{Q-f_i}+1=\int_{0}^{\infty}g(x)dx+1\leq \int_{0}^{Q-\eta}\frac{1}{Q-x} dx+1=L\]
 The first equality uses the definition of $S(t,\phi)$ and the second equality uses $f_{k+1}=Q$. 
    \end{proof}

\subsection{Proof of Lemma~\ref{lem:logp}}
     Let $T$ denote the binary tree  with $m$ leaves that minimizes the   sum of $p^{th}$ powers of the leaf-depths, for any $p\ge 1$. We will show that $T$  must be the Huffman tree. Let 
     $D_p(T):= \sum_{i\in L(T)} d(i)^p$ be the objective of tree $T$. 
    
    Note that every non-leaf node in $T$ must have 2 children. Indeed, if node $v\in T$ has exactly one child $u$ then we can remove $v$ and hang the subtree rooted at $u$ below $v$'s parent: this   reduces  the objective $D_p(T)$, contrary to the optimality of $T$.

    We now claim that the difference between the maximum/minimum depth of leaves in $T$ is at most one. Let $b$ denote a deepest leaf in $T$, and $b'$ its sibling. Note that $b'$ is  also   a leaf in $T$ because $b$'s parent (which is non-leaf) must have 2 children. Let $a\in T$ be a leaf of minimum depth.     Suppose, for a contradiction, that $d(a)\le d(b)-2$.  Let tree $T'$ be obtained from $T$ by adding two children of $a$ as leaves and removing the leaves $\{b,b'\}$ (which makes their parent a leaf in $T'$). Now,
    \begin{align*}
        D_p(T') - D_p(T) &= 2\cdot (d(a)+1)^p + (d(b)-1)^p - \left( d(a)^p + 2\cdot d(b)^p\right)\\
    &=     2\cdot \left( (d(a)+1)^p - d(b)^p\right) \,+\, (d(b)-1)^p -d(a)^p \\
    &<  (d(a)+1)^p - d(b)^p \,+\, (d(b)-1)^p -d(a)^p\\
    &= (d(a)+1)^p - d(a)^p \,+\, (d(b)-1)^p -d(b)^p    \,\,\,\le \,\,\, 0.
    \end{align*}
    The (strict) inequality uses $d(a)+1<d(b)$, and the last inequality is by convexity of $x^p$ and $d(a)+1\le b(b)-1$. This is a contradiction to the optimality of tree $T$. 

    Let $\ell$ be the maximum leaf-depth in $T$: so every leaf has depth  $\ell-1$ or $\ell$.  
    Let $r$ be the number of leaves at depth $\ell-1$. As there is no leaf at depth less than $\ell-1$ and each non-leaf node has 2 children, tree $T$ has  $2^{\ell-1} -r$ non-leaf nodes at depth $\ell-1$. This also implies that there are exactly $2\cdot(2^{\ell-1} -r)$  leaf nodes at depth $\ell$.  As $T$ has $m$ leaves in total, we must have $r + 2\cdot(2^{\ell-1} -r)=m$, which gives $r=2^\ell-m$. Finally, using the fact that $0\le r\le 2^{\ell-1}$, we get $2^{\ell-1}\le m\le 2^\ell$. So, $\ell=\lceil \log_2 m\rceil$, which means $T$ is the Huffman tree.

\end{document}